\documentclass[10pt]{amsart}
\usepackage{amsmath}
  \usepackage{graphics} %% add this and next lines if pictures should be in esp format
\usepackage{graphicx}  %% add this and next lines if pictures should be in esp format
  \usepackage{epsfig} %For pictures: screened artwork should be set up with an 85 or 100 line screen
 \usepackage[colorlinks=true]{hyperref}
  \usepackage{float}
  \usepackage{caption}
\usepackage{amsmath}
\usepackage{wrapfig}
\usepackage{graphicx}
\usepackage{subcaption}
\usepackage{amssymb}
\usepackage{diagbox} 
\usepackage{booktabs}

\usepackage{color}
\usepackage{enumitem} 
\usepackage{comment}

\newtheorem{theorem}{Theorem}[section]

\theoremstyle{definition}

\newtheorem{remark}{Remark}

\newcommand{\beq}{\begin{equation*}}
\newcommand{\eeq}{\end{equation*}}
\newcommand{\beqn}{\begin{eqnarray*}}
\newcommand{\eeqn}{\end{eqnarray*}}

\title[TYC-SIT Strategy] 
      {Combined Trojan Y Chromosome Strategy and Sterile Insect Technique to Eliminate Mosquitoes: Modelling and Analysis}

\begin{document}
\maketitle

\centerline{
\scshape  Jingjing Lyu$^{1,2,*}$,  Musong Gu$^{1}$ and Sheng Wang$^{3}$}

\medskip
{\footnotesize
% please put the address of the first author
 \leftline{$^{1}$ College Of Computer Science, Chengdu University, Chengdu,  China}
   \medskip
   \leftline{$^{2}$ Key Laboratory of Pattern Recognition and Intelligent Information Processing, Institutions}
 \leftline{\quad  of Higher Education of Sichuan Province, Chengdu University, Chengdu, China}
    \medskip
\leftline{$^{3}$ Information Development and Management Center, Chuzhou University, Chuzhou, China}
\medskip 
\leftline{$^{*}$ Correspondence: lvjingjing@cdu.edu.cn}
 }

\noindent
\quad \\

%\emph{This draft manuscript is distributed solely for purposes of scientific peer review. Its content is deliberative and predecisional, so it must not be disclosed or released by reviewers. Because the manuscript has not yet been approved for publication by the U.S. Geological Survey (USGS), it does not represent any official USGS finding or policy.}

%%The abstract of your paper
\begin{abstract}
Sterile insect technique has been successfully applied in the control of agricultural pests, however, it has a limited ability to control mosquitoes. A promising alternative approach is Trojan Y Chromosome strategy, which works by manipulating the sex ratio of a population through the introduction of feminized $YY$ supermales that guarantee male offspring. A combined Trojan Y chromosome strategy and sterile insect technique (TYC-SIT) strategy is modeled with ordinary differential equations that allow the kinetics of the female population decline of mosquitoes to be evaluated under identical modeling conditions. The dynamical analysis leads to results on both local and global stabilities of this combined  model. Optimal control analysis is also implemented to investigate the optimal mechanisms for extinction of mosquitoes. In particular, the numerical results affirm that the combined TYC-SIT enables near elimination of mosquitoes. The conclusion has great significance for pest controls.
\end{abstract}
\quad \\
\quad \\

\section{Introduction}

Mosquito-borne diseases are transmitted by mosquitoes infected with viruses, such as Zika virus, yellow fever virus, West Nile fever virus, and dengue fever virus \cite{L6}. The spread of mosquito-borne virus in humans is mainly through the bite of mosquitoes infected with the virus like Anopheles sinensis, Anopheles anthropophagus, Aedes aegypti, Aedes albopictus, Culex mosquitoes, and etc. \cite{L7,L8,L9}. Chemical treatments such as pesiticides have been implemented for many years to eradicate mosquitoes. However, environmental problems caused by excessive use of pesticides \cite{L11}, insecticide resistance \cite{L12,L13}, and combined with the lack of vaccines \cite{L14,L15}, have called for alternative environment-friendly and sustainable approaches \cite{L1}, such as radaition-based strerile insect technique (SIT) and Trojan Y chromosome strategy (TYC).

SIT works by realeasing radiation-sterilized males to an existing population, to mate with wild females so that they have no viable offspring \cite{L0,L1}. SIT has been sucessfully applied in the control of several agricultural pests such as invasive fruit flies, lepidopteran and Hemiptera \cite{L16,L17,L18}, however, it has a limited ability to control mosquitoes due to the decline in their mating competitiveness and survival \cite{L2, L3, L3-1, L3-2}. A promosing alternative approach that has been proposed for eliminating mosquitoes is TYC, in which released feminized supermales (containing two Y chromosomes) into the field to mate with wild females, resulting in a sharp sex imbalance of subsequent generations \cite{L4GutierrezTeem, L4-1,L4-2,L4-3}. Figure~\ref{fig-L1} illustrates that an equal proportion of females and males are produced in the wild, however, the offspring is guaranteed to be male when a natural female mates with a feminized $YY$ supermale. The gradual reduction in females may lead to eventual extinction of the targeted population after several generation cycles \cite{L5}. The TYC strategy is safe because it is revesible and has the advantage of targeting a specific species, thus preserving other beneficial species and avoiding non-target effects. Furthermore, no genetically engineered genes can be transferred to subsequent generations. Also, the strength of the effect can be controlled because we can decide how many supermales to be introduced to the population.

\begin{figure}[H]	
\includegraphics[width=12 cm]{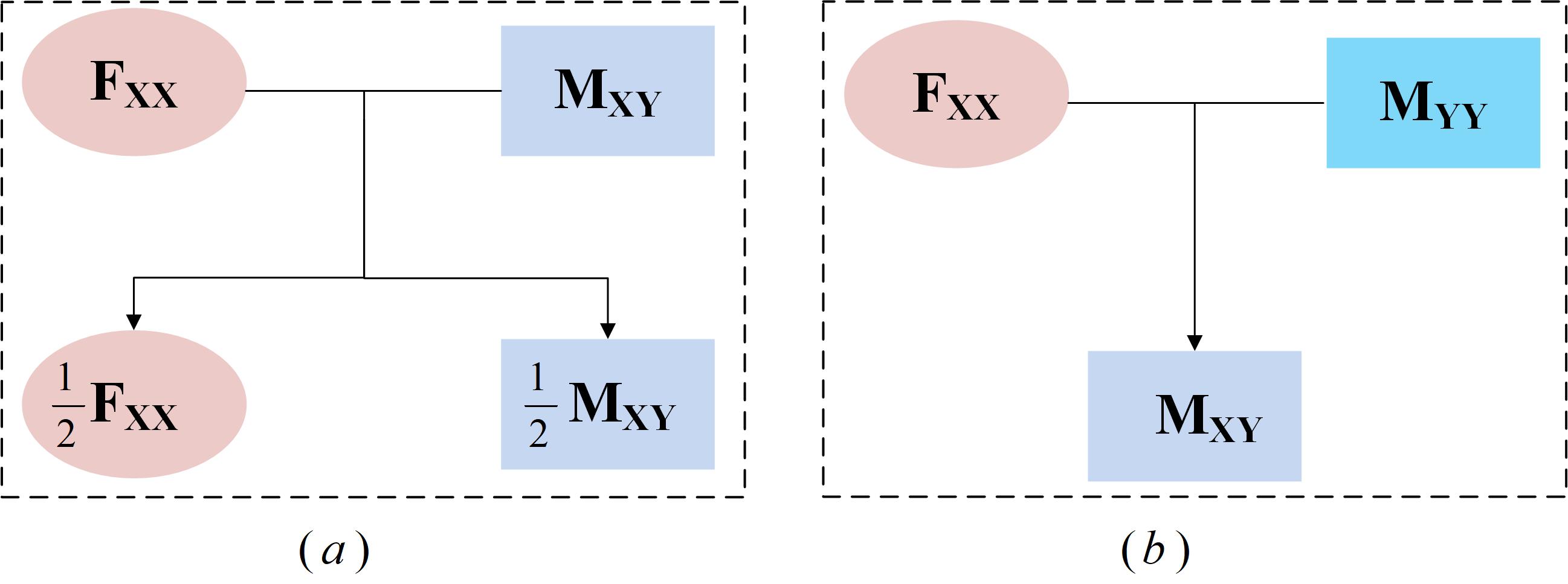}
\caption{Mating pedigrees of TYC. (\textbf{a}) Mating of a wild-type $XX$ female ($F_{XX}$) and a wild-type $XY$ male ($M_{XY}$); (\textbf{b}) Mating of a wild-type $XX$ female ($F_{XX}$) and a feminized $YY$ supermale ($M_{YY}$). \label{fig-L1}}
\end{figure} 

In this manuscript, we will compare the TYC to SIT from both dynamics and optimal control aspects. To the best of our knowledge, it is the first mathematical investigation to compare theoretical TYC to SIT.

\section{Materials and Methods}

\subsection{Mathenatical Modelling}
The population of mosquitoes in the wild tends to be numerous because their reproduction rates are high and matings occur constantly \cite{AJ}, therefore, continuous models of ordinary differential equations (ODEs) can be established to describe the population dynamics of mosquitoes. Parameters will be used in this work are firstly explained in Table \ref{Lyu-tab1}.

\begin{table}[H] 
%\tablesize{\scriptsize}
%\centering
\caption{List of Parameters. All parameters are non-negative.
\label{Lyu-tab1}}
%\tablesize{13.5cm} % You can specify the fontsize here, e.g., \tablesize{\footnotesize}. If commented out \small will be used.
\begin{tabular}{cc}
\toprule
\textbf{Parameter}	& \textbf{Description}	\\
\midrule
$\beta$		&Birth coefficient proportional to the viability of progeny			\\
$\delta$		&Death coefficient propotional to the death from predators, illness, etc.		\\
$L$		&Logistic term (to limit the size of the population) \\
$K$		&Carrying capacity of the ecosystem	\\
$\mu_1$		&Constant influx of radiation-based sterile males		\\
$\mu_2$		&Constant influx of feminized $YY$ supermales		\\
\bottomrule
\end{tabular}
\end{table}

The TYC-SIT model, in which both feminized $YY$ supermales and radiation-based sterile males are introduced, is described by a system of four ODEs for state variables: a wild-type $XX$ female ($f$), a wild-type $XY$ male ($m$), a sterile male ($s_1$), and a feminized $YY$ supermale ($s_2$). Half female and half male offspring are produced in the wild, however, only males can be produced if a wild female mates with a feminized $YY$ supermale, and no viable offspring can be produced if mates with a sterile male. Thusly, the set of equations that describe the system is:

\begin{equation}
\label{TYC-SIT}
\begin{split}
& \frac{df}{dt}=\frac{1}{2}fm\beta L \frac{m}{m+s_1+s_2}-\delta f,\\
& \frac{dm}{dt}=\frac{1}{2}fm\beta L \frac{m}{m+s_1+s_2}+f s_2 \beta L \frac{s_2}{m+s_1+s_2}-\delta m,\\
& \frac{ds_1}{dt}=\mu_1-\delta s_1,\\% s_1 represents sterile male
& \frac{ds_2}{dt}=\mu_2-\delta s_2,% s_2 represents supermales
\end{split}
\end{equation}
where $f, m, s_1, s_2$ define the number of individuals in each associated class, and 

\begin{equation}
\label{jj_1}
L=1-\frac{f+m+s_1+s_2}{K}.
\end{equation}

The intraspecies competition for female mates caused by the introduction of feminized $YY$ supermales and radiation-based sterile males is considered and modeled as 

\begin{equation}
\frac{m}{m+s_1+s_2} \quad \mbox{and} \quad \frac{s_2}{m+s_1+s_2}. 
\end{equation}
If $s_1=s_2=0$, there is no mating pressure on wild males ($\frac{m}{m+s_1+s_2}=1$). Obviously, the range of $\frac{m}{m+s_1+s_2}$ and $\frac{s_2}{m+s_1+s_2}$ is between 0 and 1. The larger the value of $\frac{m}{m+s_1+s_2}$ is, the more competitive wild males are. Similarly, the larger the value of $\frac{s_2}{m+s_1+s_2}$ is, the more competitive feminized $YY$ supermales are.

\subsection{Equilibrium and Stability Analysis}

The stability of the TYC-SIT model \eqref{TYC-SIT} is now investigated. There is one equilibrium, $(0,0, \frac{\mu_1}{\delta},\frac{\mu_2}{\delta})$, on the boundary. It's clear $(0,0, \frac{\mu_1}{\delta},\frac{\mu_2}{\delta})$ is global stable. To get the positive interior equilibrium explicitly, i.e., $(f^*,m^*,s_1^*,s_2^*)$, it is equivalent to solve the following equations:

\begin{equation}
\label{E-TYC_SIT_1}
\begin{split}
& \frac{1}{2}f^* m^* \beta L^* \frac{m^*}{m^*+s_1^*+s_2^*}-\delta f^*=0,\\
& \frac{1}{2}f^* m^* \beta L^* \frac{m^*}{m^*+s_1^*+s_2^*}+f s_2^* \beta L^* \frac{s_2^*}{m^*+s_1^*+s_2^*}-\delta m^*=0,\\
& \mu_1-\delta s_1^*=0,\\
& \mu_2-\delta s_2^*=0,\\
&L^*=1-\frac{f^*+m^*+s_1^*+s_2^*}{K}.
\end{split}
\end{equation}
By solving these equations, we have 

\begin{equation}
\label{E-TYC_SIT_2}
\begin{split}
f^* &=K-(m^*+s_1^*+s_2^*)-\frac{2\delta K (m^*+s_1^*+s_2^*)}{{m^*}^2\beta}
, \\
s^*_1&=\frac{\mu_1}{\delta},\\
s^*_2&=\frac{\mu_2}{\delta},
\end{split}
\end{equation}
and $m^*$ can be calculated from 

\begin{equation}
\begin{split}
& 2\beta {m^*}^5+\beta \left(s_1+s_2-K \right){m^*}^4+2 \left(K\delta+\beta s_2^2 \right){m^*}^3+\\
& \left[2\delta K(s_1+s_2)+2\beta s_2^2(s_1+s_2-1) \right]{m^*}^2+4\delta K *s_2^2 m^* +4K\delta s_2^2 \left(s_1+s_2 \right)=0.
\end{split}
\end{equation}
%\begin{equation}
%\begin{split}
%& f^* =\frac{{m^*}^2}{ m^*+2s_2^*} ,\\
%& m^*=\frac{\beta(K-s^*)+\sqrt{\beta^2 (K-s^*)^2-16\beta K \delta}}{4\beta},\\
%& s^*_1=\frac{\mu_3}{\delta},\\
%& s^*_2=\frac{\mu_4}{\delta}.
%\end{split}
%\end{equation}
The Jacobian matrix of the model \eqref{TYC-SIT} about $(f^*, m^*,s_1^*,s_2^*)$ is given as 

\begin{equation} \label{Jac_1} 
\mathrm{J}=\left[ \begin{array}{cccc}
J_ {11}& J_{12}  & J_{13} & J_{14} \\ 
J_{21} & J_{22} & J_{23}  & J_{24}\\ 
0 & 0 & J_{33} & 0  \\
0 & 0 &  0 & J_{44} \end{array}
\right] 
\end{equation} 
where

\begin{equation}
\begin{split}
& J_{11} = \frac{\beta K {m^*}^2 L^*-\beta f^* {m^*}^2}{2K \left(m^*+s_1^*+s_2^* \right)}-\delta \quad, \\
& J_{12}=\frac{\left(2KL^*-m^* \right)\beta f^* m^*}{2K \left(m^* + s_1^* + s_2^* \right)} - \frac{\beta f^* {m^*}^2 L^*}{2* \left(m^* + s_1^* + s_2^*\right)^2} \quad ,\\
& J_{13} = -\frac{\beta f^* {m^*}^2 L^*}{2 \left(m^* + s_1^* + s_2^*\right)^2} - \frac{\beta f^* {m^*}^2}{2K \left(m^* + s_1^* + s_2^* \right)} \quad, \\
& J_{14}=-\frac{\beta f^* {m^*}^2 L^*}{2 \left(m^* + s_1^* + s_2^* \right)^2} - \frac{\beta f^* {m^*}^2}{2K \left(m^* + s_1^* + s_2^* \right)} \quad, \\
& J_{21} =  \frac{\beta L^* \left({m^*}^2+2{s_2^*}^2 \right)}{2 \left(m^* + s_1^* + s_2^* \right)}  - \frac{\beta f^*\left({m^*}^2+2{s_2^*}^2 \right)}{2K \left(m^* + s_1^* + s_2^*\right)} \quad, \\
& J_{22} =-\frac{\beta f^* L^* \left({m^*}^2+2{s_2^*}^2 \right)}{2\left(m^* + s_1^* + s_2^* \right)^2} + \frac{\beta f^* \left(2K m^* L^*-{m^*}^2-2 {s_2^*}^2 \right)}{2 K \left(m^* + s_1^* + s_2^*\right)} -\delta \quad, \\
& J_{23} =-\frac{\beta f^* L^* \left({m^*}^2+2{s_2^*}^2\right)}{2\left(m^* + s_1^* + s_2^*\right)^2}  - \frac{b*f\left({m^*}^2+2{s_2^*}^2\right)}{2K\left(m^* + s_1^* + s_2^*\right)} \quad, \\
& J_{24}= -\frac{\beta f^* L^* \left({m^*}^2+2{s_2^*}^2\right)}{2\left(m^* + s_1^* + s_2^*\right)^2}  +  \frac{\beta f^* \left(4Ks_2^* L^*-{m^*}^2-2{s_2}^*  \right)}{2K \left(m^* + s_1^* + s_2^*\right)} \quad,\\
& J_{33} =-\delta,\\
& J_{44} =-\delta.\\
\end{split}
\end{equation}
The corresponding characteristic equation is

\begin{equation}
\lambda^4+R_{1}\lambda^3+R_{2} \lambda^2+ R_{3} \lambda+ R_{4}=0
\end{equation}
where

\begin{equation}
\label{jj_17}
\begin{split}
& R_{1}=-(J_{11}+J_{22}+J_{33}+J_{44}),\\
& R_{2}=J_{11} J_{22}+J_{11}J_{33}+J_{11}J_{44}+J_{22}J_{33}+J_{22}J_{44}+J_{33}J_{44}-J_{12}J_{21},\\
& R_{3}=(J_{33}+J_{44})J_{12}J_{21}-(J_{11}+J_{22})J_{33}J_{44},\\
& R_{4}=J_{11}J_{22}J_{33}J_{44}-J_{12}J_{21}J_{33}J_{44}.
\end{split}
\end{equation}

%According to Routh-Hurwitz criterion, $(f^*, m^*,s_1^*,s_2^*)$ is locally asymptotically stable if 
%\begin{equation}
%R_{1}>0,~ R_{3}>0,~ R_{4}>0,,~ R_{1}R_{2}R_{3}-A_{0}>R_3^2+R_1^2 R_4.
%\end{equation}

\begin{theorem}
\label{thm_1}
Let $\mu>0$. The interior equilibrium $(f^*, m^*,s_1^*,s_2^*)$ is locally asymptotically stable if 
\begin{equation}
R_{1}>0,~ R_{3}>0,~ R_{4}>0,,~ R_{1}R_{2}R_{3}-A_{0}>R_3^2+R_1^2 R_4.
\end{equation}
\end{theorem}

\begin{proof}
It follows from Routh Hurwitz stability criteria. 
\end{proof}

\begin{theorem}
\label{thm_2}
In the case of $\mu_1=\mu_2=0$, the trivial equilibrium $(0,0,0,0)$ of TYC-SIT model \eqref{TYC-SIT} is globally asymptotically stable if $\beta<\frac{\delta}{K}$.
\end{theorem}

\begin{proof} 
Consider the Lyapunov function $V(f,m,s_1,s_2)=f+m+s_1+s_2$. Note that $V(f,m,s_1,s_2) \geq 0$ due to the positivity of its solutions and $V(f,m,x)$ is radially bounded. It is left to show that $V'(f,m,s_1,s_2)<0$ for all $(f,m,s_1,s_2)\not = (0,0,0,0)$. Taking the derivative of $V(f,m,s_1,s_2)$ about $t$ yields

\begin{equation}
\begin{split}
V'(f,m,s_1,s_2)&=\frac{dV}{dt}=\frac{df}{dt}+\frac{dm}{dt}+\frac{ds_1}{dt}+\frac{ds_2}{dt}\\
&=fm\beta L \frac{m}{m+s_1+s_2}+fs_2\beta L \frac{s_2}{m+s_1+s_2}-\delta f -\delta m -\delta s_1 -\delta s_2\\
&\leq fm\beta +f s_2 \beta-\delta(f+m+s_1+s_2) \\
&(since\quad  L \leq 1, \frac{m}{m+s_1+s_2} \leq 1\quad \& \quad \frac{s_2}{m+s_1+s_2}\leq 1)\\
&\leq fm\beta+f s_2 \beta-\delta(m+s_2)\\
&= f\beta(m+s_2)-\delta(m+s_2)\\
&\leq K\beta(m+s_2)-\delta(m+s_2)\\
&\leq \left(K\beta-\delta \right)\left(m+s_2\right)
\end{split}
\end{equation}
Therefore, we only need to show $K\beta-\delta<0$ to get global stability, and this requires $\beta<\frac{\delta}{K}$, which completes the proof. 
% Therefore the SIT model \eqref{SIT} under $\mu=0$ is golablly asymptotically stable if 
%\begin{equation}
%\beta<\frac{2\delta}{K},
%\end{equation}
\end{proof}
\subsection{Optimal Control Analysis}

The goal in this subsection is to investigate the mechanisms in TYC-SIT system to lead to an optimal level of  female mosquitoes density. We assume that the influx of sterile males and feminized $YY$ supermales $\mu_1$ and $\mu_2$ are not known a priori and enter them into the system as time-dependent controls $\mu_1(t)$ and $\mu_2(t)$. In fact, male mosquitoes do not bite human and only female mosquitoes bite and spread diseases \cite{L10,L20,L21}, which implies that we don't have to kill all mosquitoes, eliminating females is enough instead. Furthermore,  We also hope the production of sterile males and feminized $YY$ supermales is minimized. Herein, the following objective function is chosen:
% an objective function that minimizes the female population and minimizes the introduction of both sterile males and YY supermales is considered:

\begin{equation}
\label{objective}
Obj(\mu_1(t),\mu_2(t))=\int_{0}^{T} \frac{1}{2}f^2+ \frac{1}{2} \mu_1(t)^2+\frac{1}{2} \mu_2(t)^2  dt
\end{equation}
where $f, m,s_1,s_2$ subject to the governing equations \eqref{TYC-SIT}. Optimal strategies are derived from the objective function, where the female population is minimized and also the introduction of both sterile males and YY supermales are minimized. We search for the optimal controls within the set $U$, which is given by 

\begin{equation}
 U=\{(\mu_1,\mu_2)|\mu_1, \mu_2 \mbox{ are measurable}, 0\leq \mu_1(t), 0 \leq \mu_2(t), t \in [0,150] \}.
\end{equation}
The goal is to find the optimal controls $(\mu_1^*(t),\mu_2^*(t))$ such that

\begin{equation}
\label{min-opt}
Obj(\mu_1^*(t),\mu_2^*(t))=\underset{\mu_1(t),\mu_2(t)} {\min} \int_{0}^{T} \frac{1}{2}f^2+ \frac{1}{2} \mu_1(t)^2+\frac{1}{2} \mu_2 (t)^2 dt
\end{equation}

The existence of $\mu_1^*(t),\mu_2^*(t)$ for minimization problem \eqref{min-opt} is guaranteed in the literatures \cite{Lit_1,Lit_2,Lit_3}.

\begin{theorem}
An optimal control $(\mu_1^*(t),\mu_2^*(t))\in U$ of the system \eqref{TYC-SIT} that minimizes the objective function $Obj$ is characterized by 
$\left(\mu_1^*(t),\mu_2^*(t)\right)=\left(\underset{t} {\max}\{0,-\lambda_3\}, \underset{t} {\max}\{0,-\lambda_4\}\right)$.
\end{theorem}
\begin{proof}
Here the Pontryagin's minimum principle is used to derive the necessary conditions on this problem. The Hamiltonian $H$ in this problem is

\begin{equation}
H=\frac{1}{2}f^2+ \frac{1}{2} \mu_1(t)^2+\frac{1}{2} \mu_2(t)^2+\lambda_1 f'+\lambda_2 m'+\lambda_3 s_1'+\lambda_4 s_2'
\end{equation}
The Hamiltonian is used to find the adjoint functions ($\lambda_i, i=1,2,3,4$),

\begin{equation}
\begin{split}
\lambda_1' =-\frac{\partial H}{\partial f}&=\lambda_1 \left(\delta + \frac{\beta m^2(f-K L)}{2 K (m + s_1 + s_2)} \right) -\lambda_2 \left( \frac{\beta L\left( m^2+2s_2^2\right)}{2(m + s_1 + s_2)} + \frac{\beta f (m^2+2s_2^2)}{2 K (m + s_1 + s_2)} \right) -f, \\
\lambda_2' =-\frac{\partial H}{\partial m}&= \lambda_1 \left(  \frac{\beta f m^2 L}{2(m + s_1 + s_2)^2} + \frac{\beta f m(m-2KL)}{2K(m + s_1 + s_2)} \right)         \\
&+\lambda_2 \left(\delta - \frac{\beta f m L}{m + s_1 + s_2} + \frac{\beta f L(m^2+2s_2^2)}{2(m + s_1 + s_2)^2} + \frac{\beta f \left(m^2+2s_2^2\right)}{2K(m + s_1 + s_2)} \right) \\
\lambda_3' =-\frac{\partial H}{\partial s_1}&= \lambda_1*\left( \frac{\beta f m^2 L}{2(m + s_1 + s_2)^2} + \frac{\beta f m^2}{2K(m + s_1 + s_2)} \right) \\
& +\lambda_2*\left( \frac{\beta f L \left(m^2+2s_2^2 \right)}{2(m + s_1 + s_2)^2} + \frac{\beta f (m^2+2s_2^2)}{2K(m + s_1 + s_2)}  \right)+\lambda_3 \delta \\
\lambda_4' =-\frac{\partial H}{\partial s_2}&= \lambda_1 \left( \frac{\beta f m^2 L}{2(m + s_1 + s_2)^2} + \frac{\beta f m^2}{2K(m + s_1 + s_2)} \right)\\
& + \lambda_2 \left( -\frac{2\beta f s_2 L}{m + s_1 + s_2} + \frac{\beta f L (m^2+2s_2^2)}{2(m + s1 + s2)^2 }+ \frac{\beta f (m^2+2s_2^2)}{2K(m + s_1 + s_2)}  \right) +\lambda_4 \delta \\
\end{split}
\end{equation}
To find the optimal $(\mu_1^*(t),\mu_2^*(t))$, minimize $H$ pointwise:

\begin{equation}
\frac{\partial H}{\partial \mu_1}=\mu_1+\lambda_3,\quad \frac{\partial H}{\partial \mu_2}=\mu_2+\lambda_4,
\end{equation}
Note that $\frac{1}{2}$ cancels with the 2 which comes from the square of the controls $\mu_1$ and $\mu_2$. Furthermore, the problem is indeed minimization as 

\begin{equation}
\begin{split}
\frac{\partial^2 H}{\partial \mu_1^2}=1>0,\\
\frac{\partial^2 H}{\partial \mu_2^2}=1>0.\\
\end{split}
\end{equation}
Hence the optimal solutions are

\begin{equation}
\begin{split}
\frac{\partial H}{\partial \mu_1}>0 & \rightarrow \mu_1^*(t)=0\\
\frac{\partial H}{\partial \mu_1}=0 & \rightarrow 0 \leq \mu_1^*(t)=-\lambda_3\\
\end{split}
\end{equation}

\begin{equation}
\begin{split}
\frac{\partial H}{\partial \mu_2}>0 & \rightarrow \mu_2^*(t)=0\\
\frac{\partial H}{\partial \mu_2}=0 & \rightarrow 0 \leq \mu_2^*(t)=-\lambda_4\\
\end{split}
\end{equation}
A compact way of writing the optimal control $(\mu_1^*(t),\mu_2^*(t))$ is

\begin{equation}
\mu_1^*(t)=\underset{t} {\max}\{0,-\lambda_3\},\quad \mu_2^*(t)=\underset{t} {\max}\{0,-\lambda_4\}.
\end{equation}
Now the proof is completed. 
\end{proof}

\subsection{Computational Methods}
The numerical simulations are investigated by MATLAB R2019b with the values of initial conditions and parameters shown in Table \ref{Lyu-tab2}. The ode15s solver was used to get numerical solutions of the combined TYC-SIT system. The TOMLAB Base Module and TOMLAB/SNOPT are also used to solve the optimal control problems of our dynamic systems.

\begin{table}[H] 
%\centering
%\tablesize{\scriptsize}
\caption{Initial Conditions. All parameters are non-negative.
\label{Lyu-tab2}}
%\tablesize{13.5cm} % You can specify the fontsize here, e.g., \tablesize{\footnotesize}. If commented out \small will be used.
\begin{tabular}{cc}
\toprule
\textbf{Initial Conditions}	& \textbf{Parameters}	\\
\midrule
$f(0)=100$		& $\beta=0.01$			\\
$m(0)=100$		&	$\delta=0.04$	\\
$s_1(0)=[0,10,20,50,100]$		&	$\mu_1=[0,5,10,15,20]$\\
$s_2(0)=[0,10,20,50,100]$		&	$\mu_2=[0,5,10,15,20]$\\
\bottomrule
\end{tabular}
\end{table}

\begin{remark}
The range of initial conditions and parameters were selected for purely  theoretical researches. The initial number of $f, m, s_1, s_2$ can represent thousands, tens of thousands, etc. And extinction in this manuscript is defined as the female population is less than 0.5.
\end{remark}

\section{Results}
The combined TYC-SIT system is modeled, and the initial conditions and parameters in Table \ref{Lyu-tab2} are utilized to observe the relative population decline of females ($f$) in response to the addition of the radiation-based sterile males and feminized $YY$ supermales.

Under the condition of  $\mu_1=0, \mu_2=0$, that is, sterile males or feminized $YY$ supermales are added to the population only once at time  $t=0$, and no additional will be introduced, it is observed that no matter how large the initial introduction of sterile males or feminized $YY$ supermales are, the system of the combined TYC-SIT cannot achieve extinction, and instead leads to an equilibrium state. As shown in Figure~\ref{fig-L2-1}, the population did decline for some time with large enough (i.e. purple/green star line) influx of sterile males or feminized $YY$ supermales, however, the population recovers soon and reaches the equilibrium state at approximately 172, which can be also calculated from equations \eqref{E-TYC_SIT_1}-\eqref{E-TYC_SIT_2}. Extinction can occur with continuous introducing modified males, two examples are provided in Figure \ref{fig-L2-2}.

\begin{figure}[H]
%\widefigure
\includegraphics[width=10 cm , clip=true, trim=0 0 0 32]{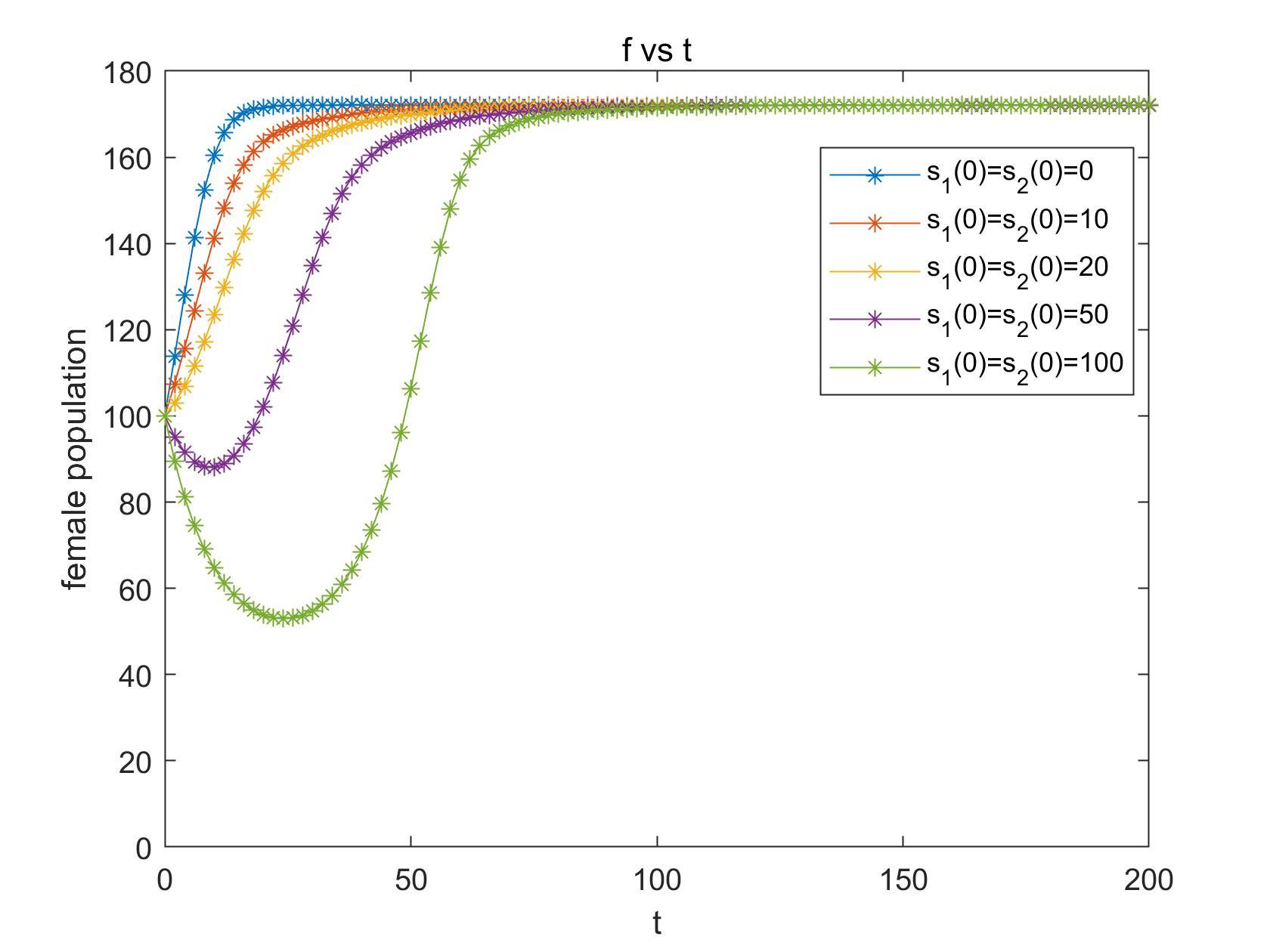}
\caption{The population change of females over time under $\mu_1=0$, $\mu_2=0$ with different number of initial introduction of sterile males and feminized $YY$ supermales. They all reach a stable state. }% (\textbf{a}) Left panel: the initial introductions of sterile males and feminized $YY$ supermales are  $s_1(0)=10, s_2(0)=10$; (\textbf{b}) Right panel: the initial introduction of sterile males and feminized $YY$ supermales are $s_1(0)=100, s_2(0)=100$. 
\label{fig-L2-1}
\end{figure}

\begin{figure}[H]
%\widefigure
\includegraphics[width=10 cm , clip=true, trim=0 0 0 32]{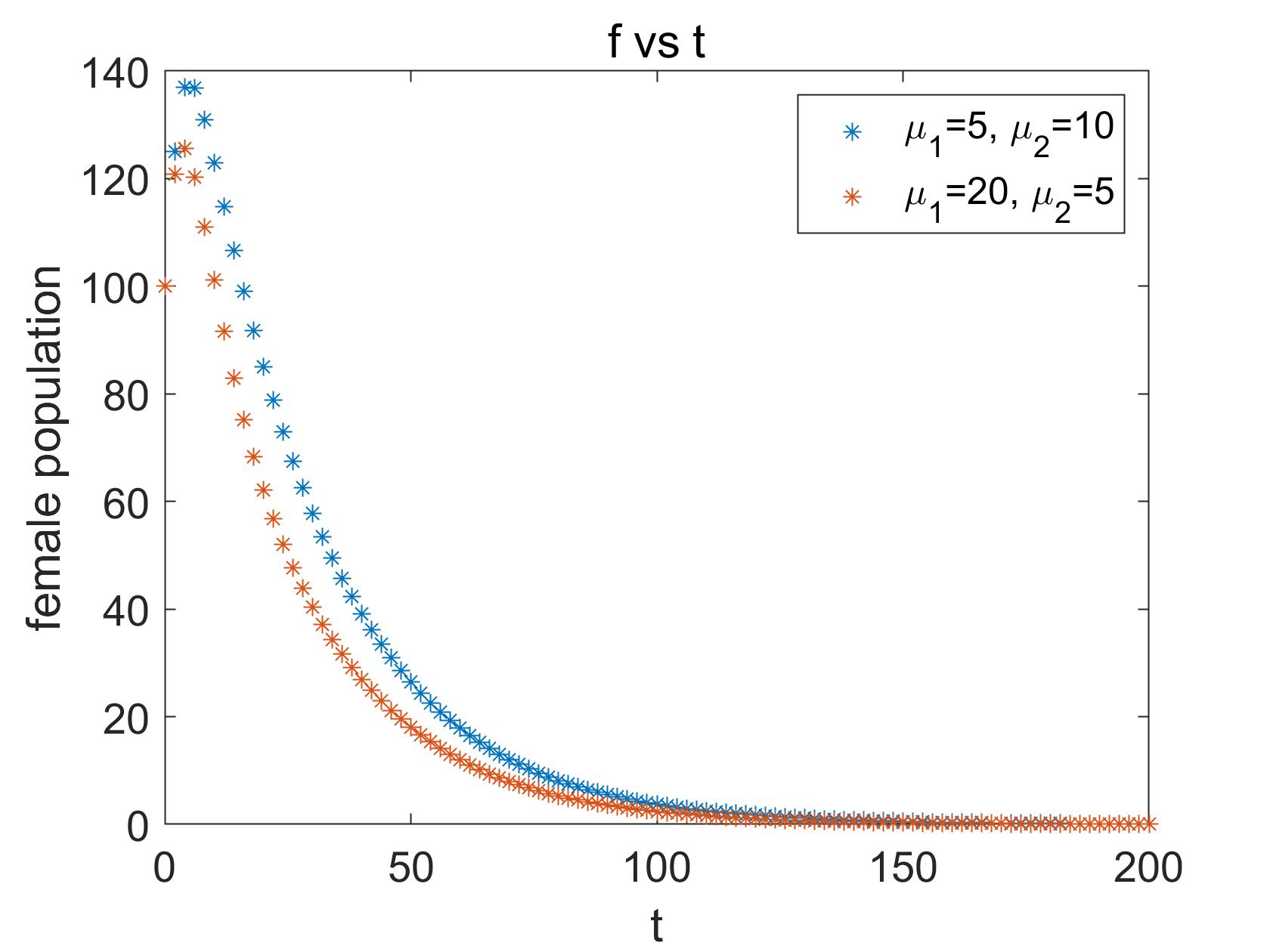}
\caption{The population change of females over time under continous constant introduction of sterile males and feminized $YY$ supermales. }
\label{fig-L2-2}
\end{figure}

If we solely continue adding sterile males (i.e. $\mu_2=0$) or solely continue adding feminized $YY$ supermales (i.e. $\mu_1=0$) after the initial introduction of the both, how and how many modified males would be introduced has a great impact on the population decline of mosquitoes. With a relatively low influx of modified males, purely introducing sterile males or feminized $YY$ supermales cannot drive the population to extinction, however, the effectiveness of the continuous introduction of $YY$ supermales is better because the population drops much more, see part (\textbf{a}) in Figure \ref{fig-L3}. As $\mu_1$ or $\mu_2$ is increased to 10, the initial decline rate with $\mu_2=10$ is faster than $mu_1=10$, but then the situation is reversed after the threshold point. As we can see in part (\textbf{c}) and (\textbf{d}) in Figure \ref{fig-L3}, the population decline rate with either $\mu_1=0$ or $\mu_2=0$ is similar if the influx of modified males is large enough. 

\begin{figure}[H]
%\widefigure
\includegraphics[width=60mm]{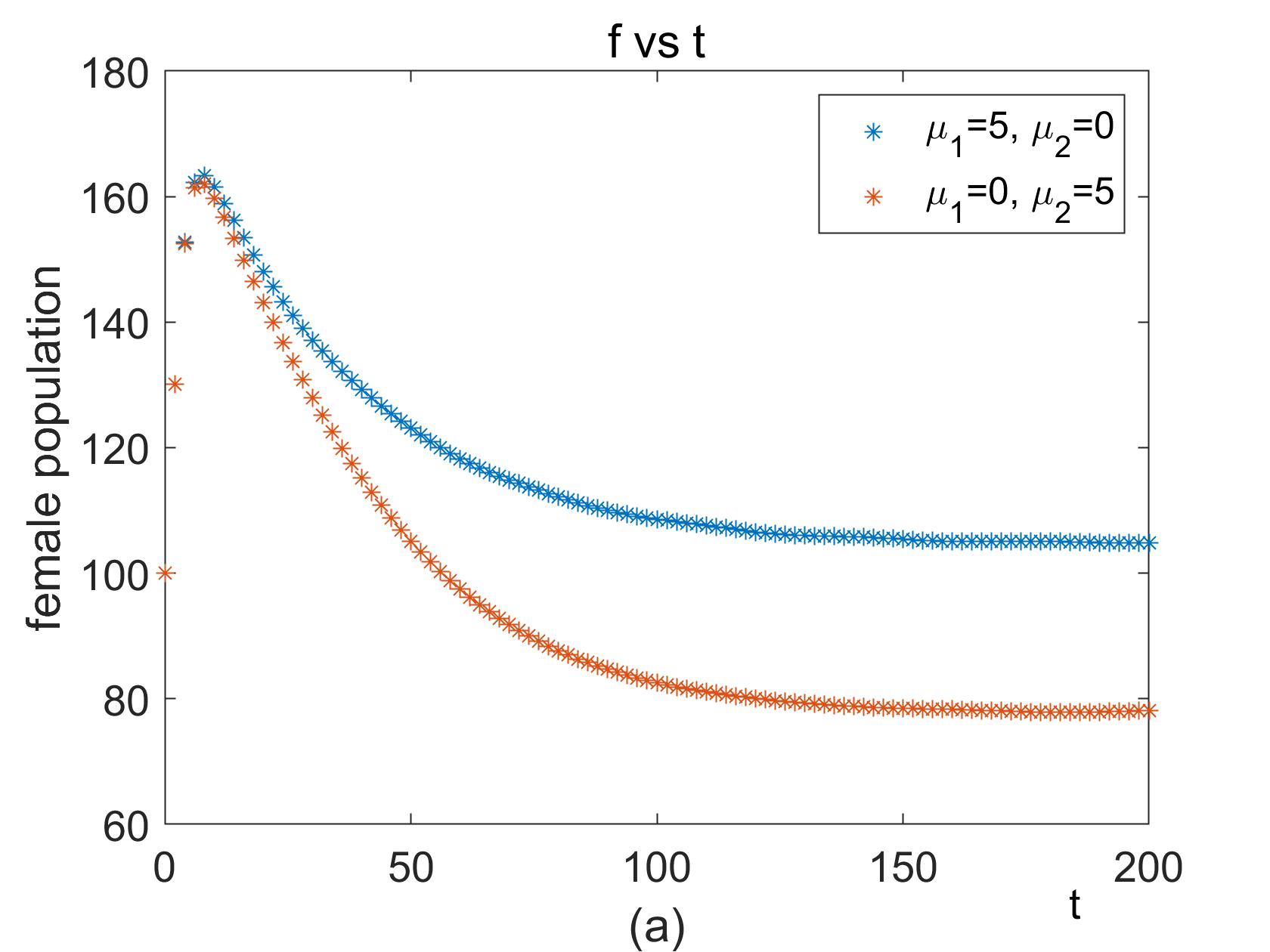}
\includegraphics[width=60mm]{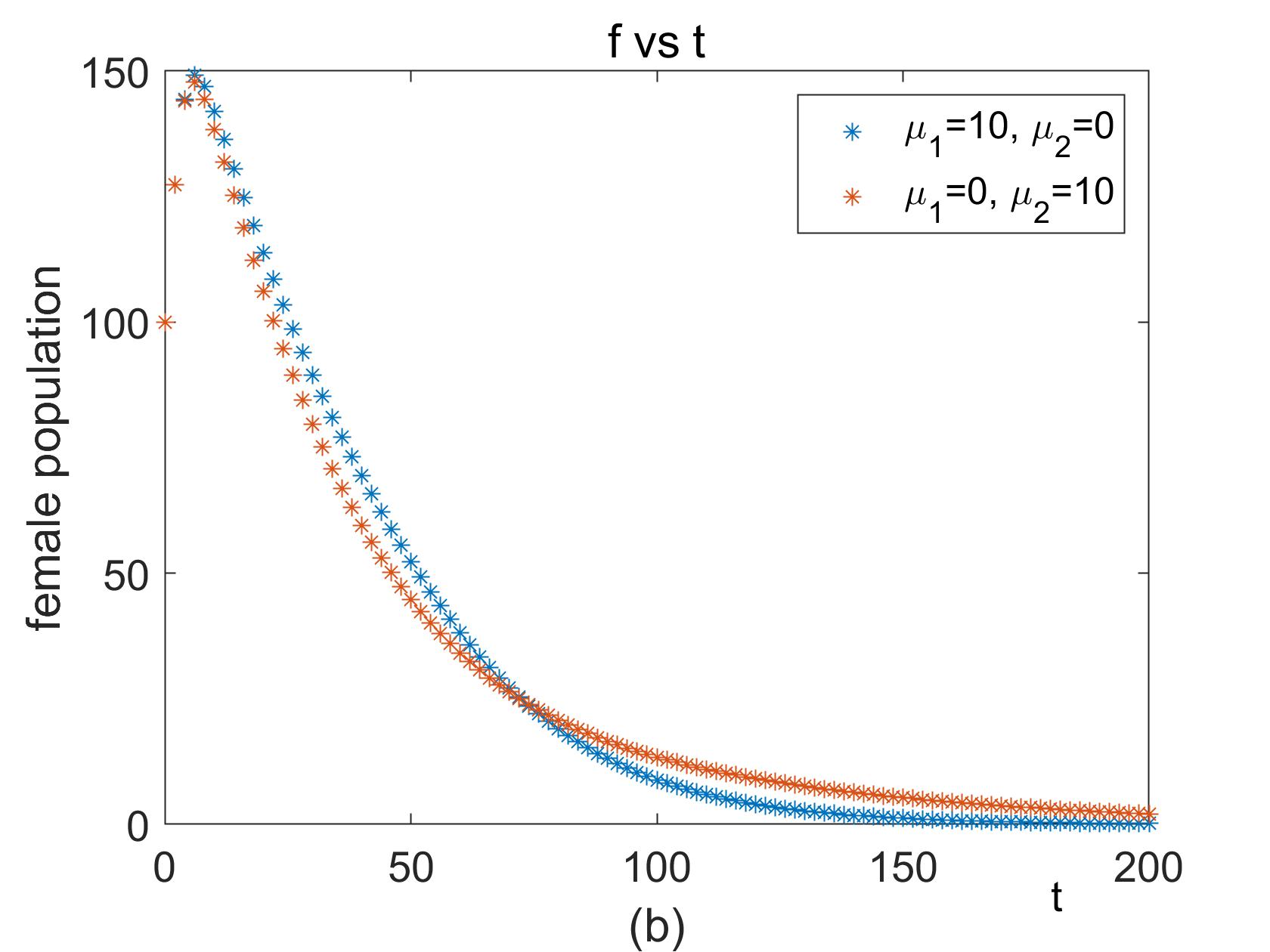} 
\includegraphics[width=60mm]{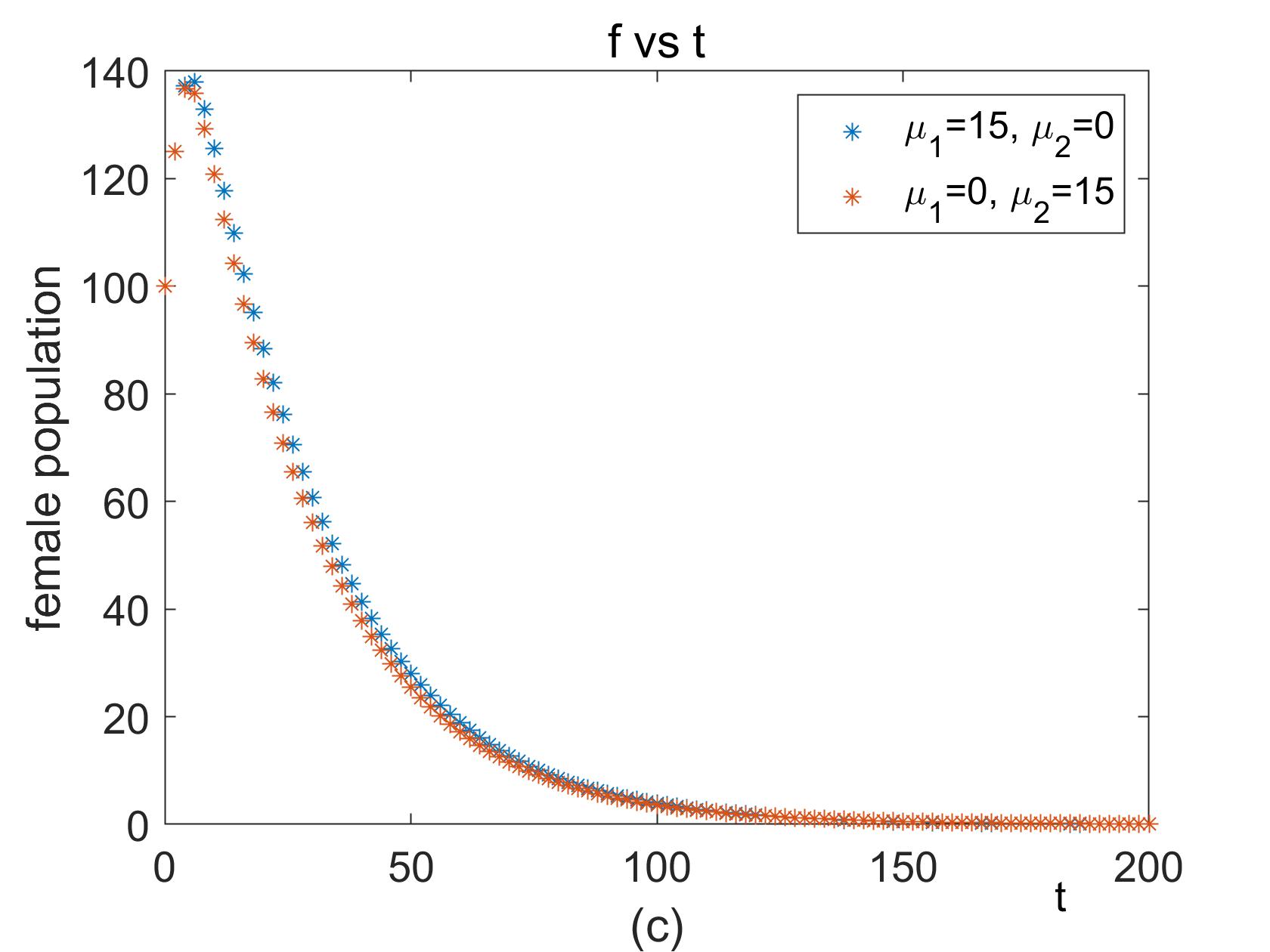}
\includegraphics[width=60mm]{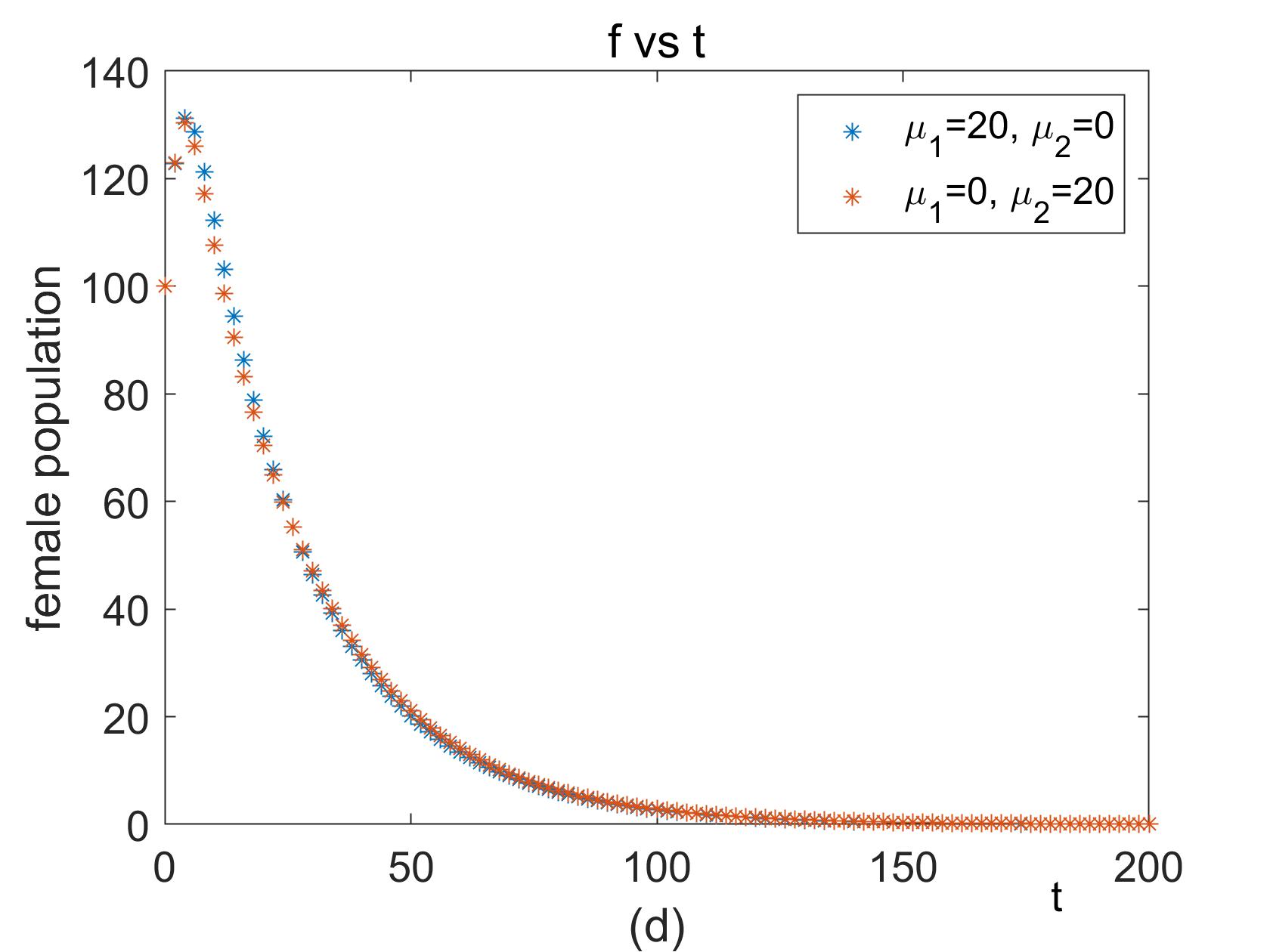}
\caption{Investigation of purely introducing sterile males or purely introducing feminized $YY$ supermales after the initial introduction of the both. Blue star line is under $\mu_2=0$, and the orange star line is under $\mu_1=0$. Here, $s_1(0)=s_2(0)=10.$ (\textbf{a}) Compared $\mu_1=5, \mu_2=0$ with $\mu_1=0, \mu_2=5$; (\textbf{b}) Compared $\mu_1=10, \mu_2=0$ with $\mu_1=0, \mu_2=10$; (\textbf{c}) Compared $\mu_1=15, \mu_2=0$ with $\mu_1=0, \mu_2=15$; (\textbf{d}) Compared $\mu_1=20, \mu_2=0$ with $\mu_1=0, \mu_2=20$. \label{fig-L3} }
\end{figure}

% optimal control analysis 
Considering the cost of the production of sterile males and feminized $YY$ supermales, optimal strategies and its corresponding optimal states have also been numerically simulated, see Figures~\ref{fig-OPT-1}-\ref{fig-OPT-2}. The population can be driven to extinction with an optimal time controls $\mu_1(t)$ or $\mu_2(t)$ by using the combined TYC-SIT approach. Figure~\ref{fig-OPT-1} shows the optimal strategy requires to introduce feminized $YY$ supermales at 20\% of the initial female population (or 5\% of the carrying capacity) for a short time to bring the population of females below some threshold, then gradually drops to a relatively low level after introducing sterile males at time $t=4$, and finally turns off until the entire population vanishes. Figure \ref{fig-OPT-2} describes the cost of implementing TYC-SIT approach varies as increasing time. After t=4,  the cost is decreasing.

\begin{figure}[H]	
%\widefigure 
\includegraphics[width=10 cm]{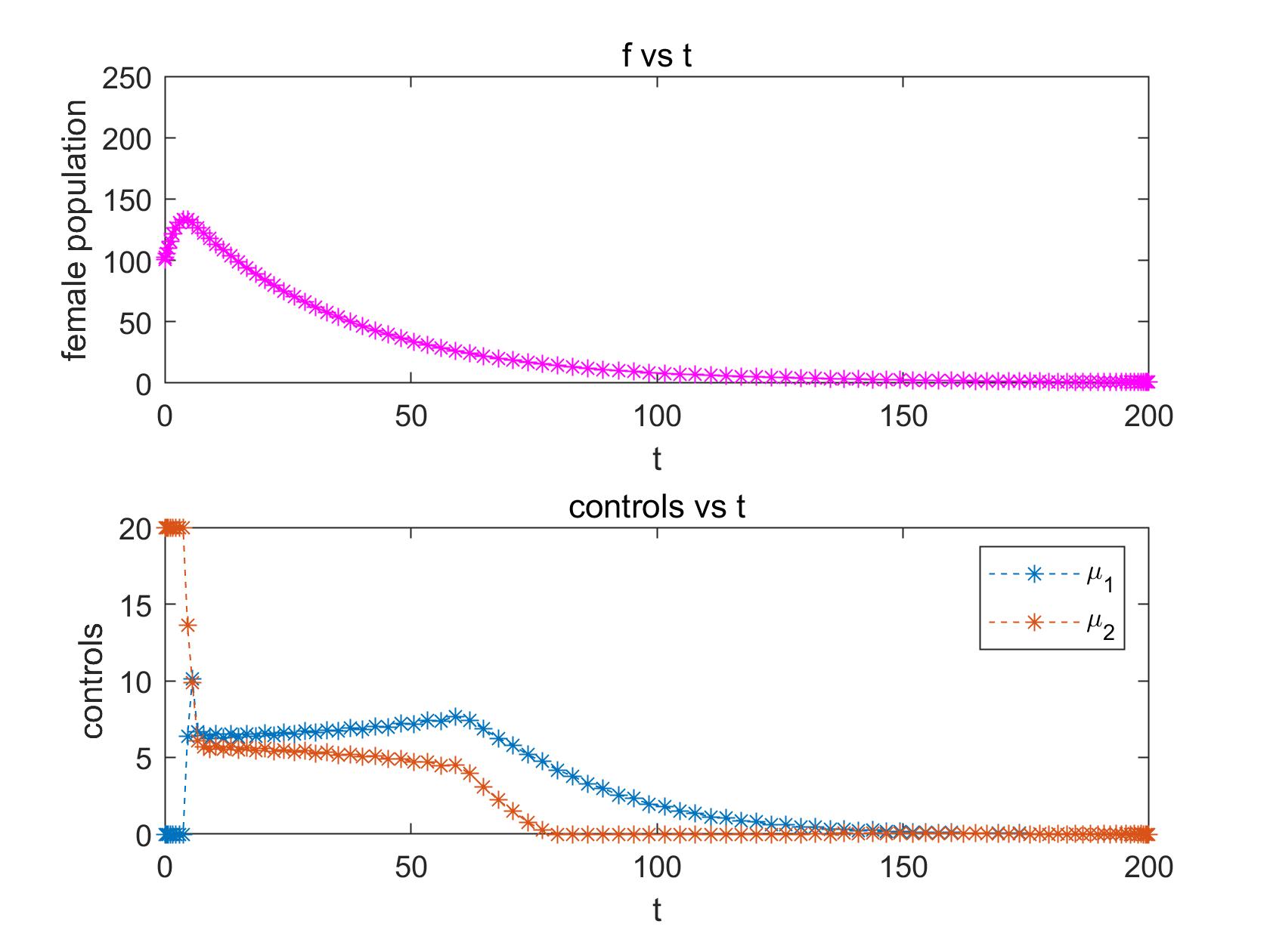}
\caption{The optimal control analysis of the combined TYC-SIT model. (\textbf{a}) Top panel: female population change with increasing time; (\textbf{b})Bottom panel: optimal controls $\mu_1$ and $\mu_2$ varies with time. \label{fig-OPT-1}}
\end{figure} 
\vspace{-2mm}
\begin{figure}[H]	
%\widefigure 
\includegraphics[width=10 cm]{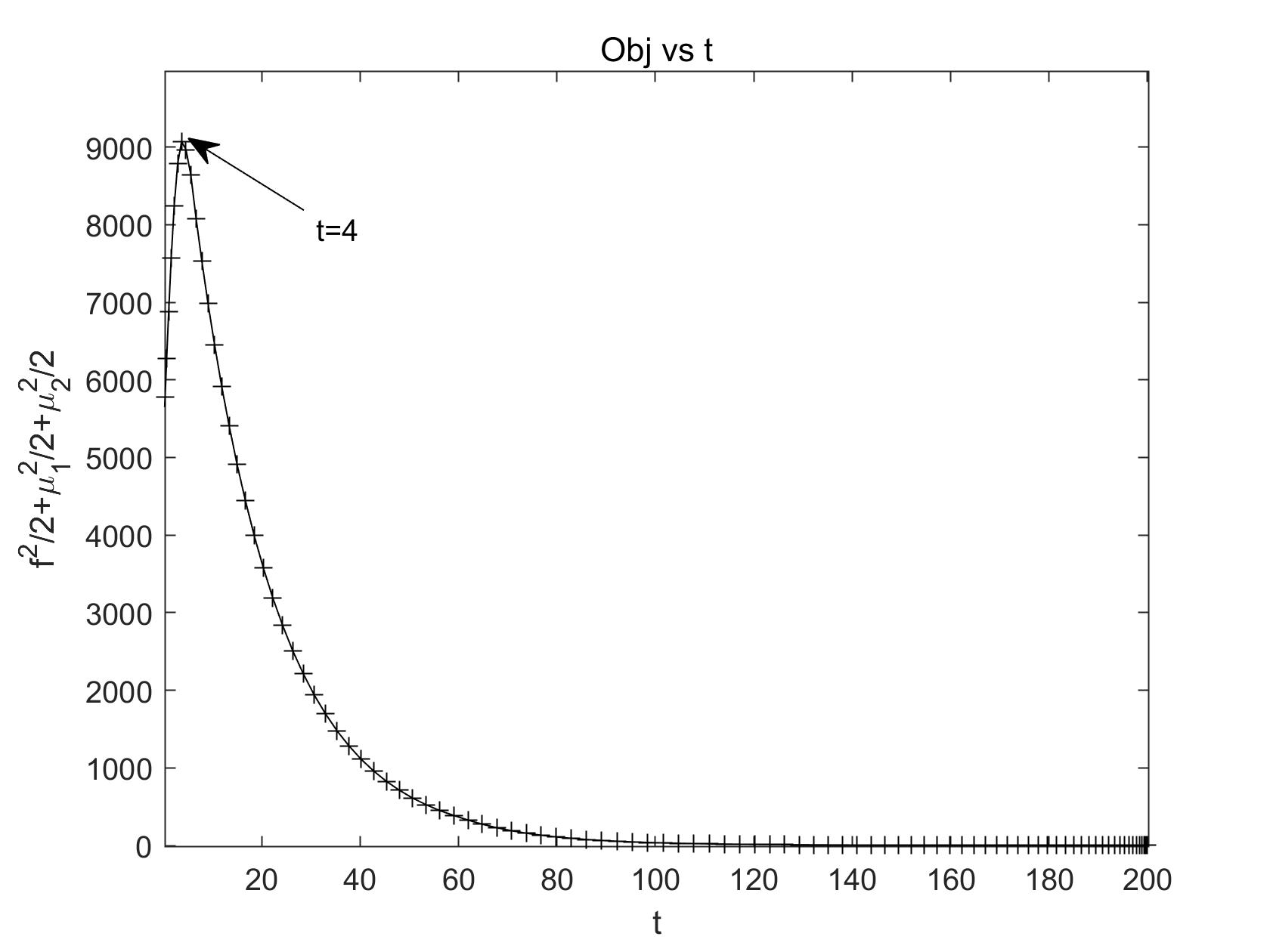}
\caption{Cost varies as increasing time with optimal controls. \label{fig-OPT-2}}
\end{figure}

\vspace{-3mm}
\section{Conclusions and Discussions}
In this manuscript, we established a mathematical model of the combined TYC and SIT approaches. The population was divided by the following four state variables: wild female mosquitoes($f$), wild male mosquitoes ($m$), radiation-based sterile males ($s_1$), and feminized $YY$ supermales ($s_2$). Six parameters, the birth coefficient $\beta$, the death coefficient $\delta$, the logistic term L, carrying capacity K, the influx of sterile males $\mu_1$, and the influx of sterile males $\mu_2$, are included. The intraspecies competition for female mates caused by the introduction of modified male mosquitoes is especially considered, which is omitted in many studies. The dynamic analysis and optimal control analysis of TYC-SIT model is important to understand the efficiency of this combined strategy to eliminate mosquitoes. Numerical simulations indicate the combined TYC-SIT approach can indeed eliminate mosquitoes. 

The combined TYC-SIT approach is safe because it is reversible and has the advantage of targeting a specific species, thus preserving other beneficial species. Furthermore, there is no genetically engineered genes can be transferred to subsequent generations. Also, the strength of the effect can be controlled because we can decide how many feminized $YY$ supermales and sterile males to be introduced to the population. Unlike other strategy, TYC-SIT does not depend on eliminating all wild matings to influence the total population. Instead, it depends on the gradual reduction in wild females over several generation cycles. These results have great significance in biological control of pests.\\

\noindent
{\bf Author Contributions:} Conceptualization, J.L., G.S. and W.S.; formal analysis, J.L.; investigation, J.L. and M.S.; writing---original draft preparation, J.L.; writing---review and editing, M.S. and W.S.; project administration, J.L., G.S. and W.S.; funding acquisition, J.L. and M.S. All authors have read and agreed to the published version of the manuscript.\\

\noindent
{\bf Funding:} This research was funded by the Key Laboratory of Pattern Recognition and Intelligent Information Processing, Institutions of Higher Education of Sichuan Province (Grant $\#$: MSSB-2021-09) and Talent Initiation Program of Chengdu University (Grant $\#$: 2081921002).\\

\noindent
{\bf Data Availability Statement:} The data used to support the findings of this study are available from the corresponding author upon request.\\

\noindent
{\bf Acknowledgements:}  The authors thank Dr. Chuan He for his support throughout this work. We are also grateful to two anonymous reviewers for their valuable suggestions.\\

\noindent
{\bf Conflicts of Interest:} The authors declare no conflict of interest.\\

\noindent
{\bf Abbreviations\\} 
\noindent
The following abbreviations are used in this manuscript:\\
\noindent
\begin{tabular}{@{}ll}
SIT & Sterile Insect Technique\\
TYC &Trojan Y Chromosome Strategy\\
ODEs &  Ordinary Differential Equations
\end{tabular}

\end{document}